\documentclass[11pt]{article}
\usepackage[letterpaper, margin=1in, headheight=105pt]{geometry}
\usepackage{subcaption}
\usepackage{enumerate}
\usepackage{mathtools}
\usepackage{listings}
\usepackage{calligra}
\usepackage{fancyhdr}
\usepackage{titlesec}
\usepackage{graphicx}
\usepackage{amsmath}
\usepackage{amssymb}
\usepackage{caption}
\usepackage{titling}
\usepackage{comment}
\usepackage{xparse}
\usepackage{amsthm}
\usepackage{bbm}
\usepackage[dvipsnames]{xcolor}
\usepackage[colorlinks=true,
			linkcolor=WildStrawberry,
			anchorcolor=red,
			citecolor=green,
			urlcolor=RoyalBlue,
			filecolor=brown,
			menucolor=pink]{hyperref}
\usepackage[numbered,
			framed]{matlab-prettifier}

\newcommand{\ignore}[1]{}

\usepackage[title]{appendix}

\usepackage{tikz}
\usetikzlibrary{positioning}
\tikzset{plain/.style={circle,fill=blue!20,draw,minimum size=0.5cm,inner sep=0pt},
}
\usepackage{hyperref}

\newtheorem{theorem}{Theorem}
\newtheorem{lemma}{Lemma}

\newtheorem{corollary}{Corollary}
\theoremstyle{definition}
\newtheorem{definition}{Definition}
\newtheorem{example}{Example}

\theoremstyle{remark}
\newtheorem*{remark}{Remark}

\graphicspath{{./graphics/}}

\usepackage{accents}

\usepackage{mathtools}

\usepackage{algorithm}
\usepackage{algpseudocode}
\usepackage{ stmaryrd }

\usepackage[latin1]{inputenc}
\usepackage{tikz}
\usetikzlibrary{shapes,arrows}
\tikzstyle{decision} = [diamond, draw, fill=blue!20, 
    text width=4.5em, text badly centered, node distance=3cm, inner sep=0pt]
\tikzstyle{block} = [rectangle, draw, fill=blue!20, 
    text centered, rounded corners, minimum height=4em]
\tikzstyle{line} = [draw, -latex']
\tikzstyle{cloud} = [draw, ellipse,fill=red!20, node distance=3cm,
    minimum height=2em]

\title{Optimal Protocols for 2-Party Contention Resolution}
\author{Dingyu Wang}

\begin{document}
\maketitle

\begin{abstract}
\emph{Contention Resolution} is a fundamental symmetry-breaking problem 
in which $n$ devices must acquire
temporary and exclusive access to some \emph{shared resource},
without the assistance of a mediating authority.  For example,
the $n$ devices may be sensors that each need to transmit a single packet of data over a broadcast channel.  In each time step, devices
can (probabilistically) choose to acquire the resource or remain idle; if exactly one device attempts to acquire it, it succeeds,
and if two or more devices make an attempt, none succeeds.  The complexity of the problem depends heavily on what types of \emph{collision detection} are available.  In this paper we consider \emph{acknowledgement-based protocols}, in which devices \underline{only} learn whether their own attempt succeeded or failed; they receive no other feedback from the environment whatsoever, 
i.e., whether other devices attempted to acquire the resource, succeeded, or failed.

Nearly all work on the Contention Resolution problem evaluated the performance of algorithms \emph{asymptotically}, 
as $n\rightarrow \infty$.  In this work we focus on the simplest case of $n=2$ devices, but look for \underline{\emph{precisely}} optimal algorithms.  We design provably optimal algorithms under three natural cost metrics: minimizing the expected average of the waiting times ({\sc avg}), the expected waiting time until the first device acquires the resource ({\sc min}), and the expected time until the last device acquires the resource ({\sc max}).  
We first prove that the optimal algorithms for $n=2$ are \emph{periodic} in a certain sense, and therefore have finite descriptions, then we design optimal algorithms under all three objectives.

\begin{itemize}
    \item[{\sc avg}.] The optimal contention resolution algorithm under the {\sc avg} objective has expected cost $\sqrt{3/2} + 3/2 \approx 2.72474$. 
    \item[{\sc min}.] The optimal contention resolution algorithm under the {\sc min} objective has expected cost $2$.  (This result can be proved in an ad hoc fashion, and may be considered folklore.)
    \item[{\sc max}.] The optimal contention resolution algorithm under the {\sc max} objective has expected cost $1/\gamma \approx 3.33641$, where $\gamma\approx 0.299723$ is the smallest root of $3x^3 - 12x^2 + 10x -2$.\footnote{We may also express $\gamma$ in radical form: $\gamma = -\frac{1}{6} \left(1-i \sqrt{3}\right) \sqrt[3]{13+i \sqrt{47}}+\frac{4}{3}-\frac{1+i \sqrt{3}}{\sqrt[3]{13+i \sqrt{47}}}$.}
\end{itemize}

\end{abstract}

\section{Introduction}\label{sect:introduction}

The goal of a contention resolution scheme is to allow multiple devices to eventually 
obtain exclusive access to some shared resource.  In this paper\footnote{This is the full version of \cite{wang2021contention}. } we will use often use 
the terminology
of one particular application, namely, 
wireless devices that wish to broadcast messages on a multiple-access channel.
However, contention resolution schemes are used in a variety of areas~\cite{metcalfe1976ethernet,rajwar2001speculative,jacobson1988congestion}, 
not just wireless networking.  We consider a model of contention resolution that is distinguished by the following features.
\begin{description}
    \item[Discrete Time.] Time is partitioned into discrete \emph{slots}.  It is the goal of every device to obtain exclusive access to the channel for exactly one slot, after which it no longer participates in the protocol.  
    We assume that all $n$ devices begin at the same time, and therefore agree on slot zero.
    (Other work considers an infinite-time model in which devices are injected adversarially~\cite{BenderFGY19,ChangJP19,AwerbuchRS08}, 
    or according to a Poisson distribution~\cite{MoselyH85,TsybakovM78} 
    with some constant mean.)
    \item[Feedback.] At the beginning of each time slot each device can choose to either transmit its message or remain idle.  If it chooses to idle, it receives no feedback from the environment; if it chooses to transmit, it receives a signal indicating whether the transmission was successful (all other devices remained idle).  
    (``Full sensing'' protocols like~\cite{MoselyH85,TsybakovM78,BenderFGY19,ChangJP19,AwerbuchRS08}, 
    in contrast, depend on receiving ternary feedback at each time slot indicating whether there was no transmission, 
    some successful transmission, or a collision.)
    \item[Noiseless operation.] The system is errorless; 
    there is no environmental noise.
    \item[Anonymity.] Devices are indistinguishable and run the same algorithm, but can break symmetry by generating (private) random bits.
\end{description}

There are many ways to measure the time-efficiency of 
contention resolution protocols.
In infinite-time models, we want to avoid deadlock~\cite{Aldous87,BenderFHKL04,BenderFHKL05,BenderFGY19,BenderKPY18,ChangJP19}, 
minimize the latency of devices
in the system, and generally make productive use of a (large) 
constant fraction of the slots~\cite{BenderFGY19,BenderKPY18,ChangJP19}.
When all $n$ devices begin at the same time~\cite{BenderFHKL04,BenderFHKL05}, there are still several
natural measures of efficiency.  In this paper we consider three: minimizing
the time until the \emph{first} successful transmission ({\sc min}), 
the \emph{last} successful transmission time ({\sc max}, a.k.a.~the \emph{makespan}), and the \emph{average} transmission time ({\sc avg}).

\subsection{Prior Work}
Classic infinite-time protocols like ALOHA \cite{abramson1970aloha} 
and binary exponential backoff algorithms~\cite{MetcalfeB76,kwak2005performance} 
are simple but suffer from poor worst case performance 
and eventual deadlock~\cite{Aldous87,BenderFHKL04,BenderFHKL05}, even
under \emph{non}-adversarial injection rates, e.g., Poisson
injection rates with arbitrary small means.
These are \emph{acknowledgement-based} protocols which do not require
constant (ternary) channel feedback.
One line of work aimed to achieve deadlock-freeness 
under Poisson arrivals~\cite{Capetanakis79,Gallager78,MoselyH85,TsybakovM78},
assuming ternary channel feedback.
The maximum channel usage rate is known to be between 0.48776~\cite{MoselyH85,TsybakovM78}
and 0.5874~\cite{TsybakovM81}.  A different line of work aimed at achieving deadlock-freeness and constant rate of efficiency under \emph{adversarial} injections and possibly adversarial \emph{jamming}, also assuming ternary feedback.  See \cite{BenderFGY19,ChangJP19,AwerbuchRS08} for robust 
protocols that can tolerate a jamming adversary.
One problem with both of these lines of work is that all devices must monitor
the channel constantly (for the ternary silence/success/collision feedback).  
Bender et al.~\cite{BenderKPY18} considered adversarial injection rates
and showed that it is possible to achieve a constant efficiency rate 
while only monitoring/participating in $O(\log(\log^* n))$ time slots.
This was later shown to be optimal~\cite{ChangKPWZ17}.

When all $n$ devices start at the same time slot ($n$ unknown), 
we have a pretty 
good understanding of the {\sc avg}, {\sc min}, and {\sc max} objectives.
Here there are still variants of the problem, 
depending on whether the protocol is 
full-sensing (requiring ternery feedback) or merely acknowledgement-based.
Willard~\cite{Willard86} and Nakano and Olariu~\cite{NakanoO02} 
gave full sensing protocols for the {\sc min} objective
when $n$ is unknown that takes time $O(\log\log n + \log f^{-1})$ with probability $1-f$, which is optimal.
The \emph{decay} algorithm~\cite{Bar-YehudaGI92} is an acknowledgement-based protocol
for the {\sc min} objective that runs in $O(\log n\log f^{-1})$ time with probability $1-f$,
which is also known to be optimal~\cite{Newport14}.  When $n$ is unknown,
binary exponential backoff achieves optimal $O(n)$ time under the {\sc avg} objective,
but suboptimal $\Theta(n\log n)$ time under the {\sc max} objective~\cite{BenderFHKL04,BenderFHKL05}.
The \emph{sawtooth} protocol of 
Bender et al.~\cite{BenderFHKL04,BenderFHKL05} is optimal $O(n)$ under
both {\sc avg} and {\sc max}; it is acknowledgement-based.

\subsection{New Results}

In this paper we consider what seems to be the \emph{simplest non-trivial symmetry breaking problem}, namely, resolving contention among two parties ($n=2$) via an acknowledgement-based protocol.  The \emph{asymptotic} complexity of this problem is not difficult to derive:
$O(1)$ time suffices, under any reasonable objective function, and $O(\log f^{-1})$ time 
suffices with probability $1-f$.  However, our goal is to discover 
\underline{\emph{precisely}} optimal algorithms.

We derive the optimal protocols for the {\sc avg}, {\sc min}, and {\sc max} objectives, in expectation, which are produced below.  The optimal {\sc min} protocol is easy to obtain using \emph{ad hoc} arguments;
it has expected cost $2$.  However, the optimal protocols for {\sc avg} and {\sc max} require a more principled, rigorous approach to the problem.
We show that the protocol minimizing {\sc avg} has expected cost $\sqrt{3/2}+3/2\approx 2.72474$,
and that the optimal protocol minimizing {\sc max} has expected
cost $1/\gamma \approx 3.33641$, 
where $\gamma\approx 0.299723$ is the unique root of 
$3x^3 - 12x^2 + 10x -2$ in the interval $[1/4,1/3]$.

\begin{center}

\framebox{
\begin{minipage}{6.3in}
{\bf {\sc avg}-Contention Resolution:}
\begin{description}
    \item[Step 1.] Transmit with probability $\frac{4-\sqrt{6}}{3}\approx 0.516837$.  If successful, halt; if there was a collision, repeat Step 1; otherwise proceed to Step 2.
    \item[Step 2.] Transmit with probability $\frac{1+\sqrt{6}}{5}\approx 0.689898$.  If successful, halt; if there was a collision, go to Step 1; otherwise proceed to Step 3.
    \item[Step 3.] Transmit with probability 1. If successful, halt; otherwise go to Step 1.
\end{description}
\end{minipage}
}

\medskip

\framebox{
\begin{minipage}{6.3in}
{\bf {\sc min}-Contention Resolution:}
\begin{itemize}
    \item In each step, transmit with probability 1/2 until successful.
\end{itemize}
\end{minipage}
}

\medskip

\framebox{
\begin{minipage}{6.3in}
{\bf {\sc max}-Contention Resolution:}
\begin{description}
    \item[Step 1.] Transmit with probability $\alpha\approx 0.528837$, where $\alpha$ is the unique root 
    of $x^3 + 7x^2 - 21x + 9$ in $[0,1]$. 
    If successful, halt; if there was a collision, repeat Step 1; otherwise proceed to Step 2.
    \item[Step 2.] Transmit with probability $\beta\approx 0.785997$, where $\beta$ is the unique root of $4x^3 - 8x^2 + 3$ in $[0,1]$. 
    If successful, halt; if there was a collision, go to Step 1; otherwise proceed to Step 3.
    \item[Step 3.] Transmit with probability 1. 
    If successful, halt; otherwise go to Step 1.
\end{description}
\end{minipage}
}
\end{center}

One may naturally ask: what is the point of understanding Contention Resolution problems with $n=O(1)$ devices?  The most straightforward answer is that
in some applications, contention resolution instances between
$n=O(1)$ devices are commonplace.\footnote{For a humorous example, 
consider the Canadian Standoff problem \url{https://www.cartoonstock.com/cartoonview.asp?catref=CC137954}.}
However, even if one is only interested in the asymptotic case of
$n\rightarrow \infty$ devices, understanding how to resolve $n=O(1)$ 
optimally is essential.  For example, the protocols of~\cite{Capetanakis79,Gallager78,MoselyH85,TsybakovM78} work by
repeatedly isolating subsets of the $n'$ active devices, where $n'$
is Poisson distributed with mean around 1.1, then resolving conflicts
within this set (if $n'>1$) using a near-optimal procedure.
The channel usage rate of these protocols ($\approx 0.48776$) depends
critically on the efficiency of Contention Resolution 
among $n'$ devices, where $\mathbb{E}[n']=O(1)$.
Moreover, \emph{improving} these algorithms will likely require a much
better understanding of $O(1)$-size contention resolution.

\paragraph{Organization.}
In Section~\ref{sect:problemformulation} we give a formal
definition of the model and state Theorem~\ref{thm:existence}
on the \emph{existence} of an optimal protocol for any reasonable
objective function.  In Section~\ref{sect:twoparty} we prove another structural
result on optimal protocols for $n=2$ devices under the {\sc avg}, {\sc min}, and {\sc max} objectives (Theorem~\ref{thm:recurrent}), 
and use it to characterize what the optimal protocols 
for {\sc avg} (Theorem~\ref{thm:min_average}), {\sc min} (Theorem~\ref{thm:min_short}),
and {\sc max} (Theorem~\ref{thm:min_long}) should look like.  
Corollary~\ref{cor:average} derives that {\bf {\sc avg}-Contention Resolution} is the optimal protocol under the {\sc avg} objective,
and Corollary~\ref{cor:long} does the same for {\bf {\sc max}-Contention Resolution} under {\sc max}.
The proofs of Theorems~\ref{thm:existence} and \ref{thm:recurrent}
and Corollaries~\ref{cor:average} and~\ref{cor:long} appear in the Appendix.

\section{Problem Formulation}\label{sect:problemformulation}

After each time step the channel issues responses to the devices from the set
$\mathcal{R} = \{0,1,2_+\}$.  If the device idles, it always receives $0$. 
If it attempts to transmit, it receives $1$ if successful and $2_+$ if unsuccessful.
A \emph{history} is word over $\mathcal{R}^*$.  We use exponents for repetition and $*$ as short for $\mathcal{R}^*$; e.g., the history $0^3 2_+^2$ is short for $0002_+2_+$ and $* 1 *$ is the set of all histories containing a $1$. 
The notation $a\in w$ means that symbol $a$ has at least one occurrence in word $w$.

Devices choose their action (transmit or idle) 
at time step $t \in \mathbb{N}$ and receive feedback at time $t+0.5$.
A \emph{policy} is a function $f$ for deciding the probability of 
transmitting.
Define $\mathcal{F}=\{f:\mathcal{R}^* \to [0,1]\mid \forall w\in\mathcal{R}^*, 1\in w\implies f(w)=0\}$ to be the set of all proper policies, i.e.,
once a device is successful ($1\in w$), it must halt ($f(w)=0$).\footnote{A policy may have no finite representation, and therefore may not be an \emph{algorithm} in the usual sense.}
Every particular policy $f\in\mathcal{F}$ induces a distribution
on decisions $\{D_{k,t}\}_{k\in [n], t\in\mathbb{N}}$ 
and responses $\{R_{k,t}\}_{k\in[n],t\in \mathbb{N}}$, 
where $D_{k,t}=1$ iff the $k$th device transmits at time $t$
and $R_{k,t} \in \mathcal{R}$ is the response received by the $k$th device
at time $t+0.5$.  In particular,
\begin{align}
&\mathbb{P}(D_{k,t}=1\mid R_{k,0}R_{k,1}\cdots R_{k,t-1}=h)=f(h)\label{eq:send_prob},\\
&R_{k,{t}}(w)=\begin{cases}
0,&D_{k,t}(w)=0\\
1,&(D_{k,t}(w)=1) \land (\forall j\neq k, D_{j,t}(w)=0)\\
2^+,&(D_{k,t}(w)=1) \land (\exists j\neq k, D_{j,t}(w)=1)
\end{cases}\label{eq:response}
\end{align}
Define $X_i$ to be the random variable of the number of time
slots until device $i$ succeeds.  Note that since we number the slots
starting from zero, 
\begin{align*}
X_i &= 1+\min \{t\geq 0\mid R_{i,t} = 1\}.
\end{align*}
Note that $\{X_i\}_{i\in [n]}$ are identically distributed 
but not independent.  For example, minimizing the average
of $\{X_i\}_{i\in [n]}$ is equivalent to minimizing $X_1$ since:
\begin{align*}
    \mathbb{E}\frac{\sum_{i=1}^n X_i}{n}=\frac{\sum_{i=1}^n\mathbb{E}X_i}{n}=\frac{n\mathbb{E}X_1}{n}=\mathbb{E}X_1.
\end{align*}

\subsection{Performance Metrics and Existence Issues}

For our proofs it is helpful to assume the existence of an 
\emph{optimal protocol}
but it is not immediate that there \emph{exists} such 
an optimal protocol.  (Perhaps there is just an infinite succession 
of protocols, each better than the next.)  In Appendix \ref{appendix:proof_existence}
we prove that optimal protocols exist for all ``reasonable'' objectives.
A \emph{cost function} $T : \mathbb{Z}_+^n \rightarrow \mathbb{R}_+$
is one that maps the vector of device latencies to a single (positive) cost.
The objective is to minimize $\mathbb{E}T(X_1,\ldots,X_n)$.

\begin{definition}[Informal] \label{def:informal_reasonable}
A function $T:\mathbb{Z}_+^n\to \mathbb{R}^+$ is \emph{reasonable} if for any $s>0$ there exists some $N>0$ such that $T(x_1,\ldots,x_n)<s$ 
can be known if each of $x_1,x_2,\ldots,x_n$ is either known 
or known to be greater than $N$.
\end{definition}

For example, $T_1(x_1,\ldots,x_n)=\frac{\sum_{k=1}^n x_k}{n}$ ({\sc avg}),
$T_2(x_1,\ldots,x_n)=\min(x_1,\ldots,x_n)$ ({\sc min}),
and $T_3(x_1,\ldots,x_n)=\max(x_1,\ldots,x_n)$ ({\sc max})
are all reasonable, as are all $\ell_p$ norms, etc.

\begin{theorem}\label{thm:existence}
Given the number of users $n$ and a reasonable objective function $T$, 
there exists an optimal policy $f^* \in \mathcal{F}$ that minimizes $\mathbb{E}T(X_1,X_2,\ldots,X_n)$.
\end{theorem}

\section{Contention Resolution Between Two Parties}\label{sect:twoparty}

In this section we restrict our attention to the case $n=2$.
One key observation that makes the $n=2$ case special
is that whenever one device receives $2_+$ (collision) feedback, it knows
that its history and the other device's history are identical.
For many reasonable objective functions the best response
to a collision is to restart the protocol.  
This is proved formally in Theorem~\ref{thm:recurrent} for 
a class of objective functions that includes {\sc avg}, {\sc min}, and {\sc max}.
See Appendix \ref{appendix:proof_recurrent} for proof.

\begin{theorem}\label{thm:recurrent}
Let $n=2$, $T$ be a reasonable objective function, 
and $f$ be an optimal policy for $T$.
Another policy $f^*$ is defined as follows.
\begin{align*}
    f^*(0^k)&=f(0^k),\quad \forall k\in \mathbb{N}\\
    f^*(*2_+0^k)&=f(0^k),\quad \forall k\in \mathbb{N}\\
    f^*(*1*)&=0
\end{align*}
If $T(x+c,y+c)=T(x,y)+c$ for any $c$ (scalar additivity), 
then $f^*$ is also an optimal policy for $T$.
\end{theorem}

Theorem \ref{thm:recurrent} tells us that for the objectives that are scalar additive
(including {\sc avg}, {\sc min}, and {\sc max}), 
we can restrict our attention to policies $f\in \mathcal{F}$
defined by a vector of probabilities $(p_i)_{i\ge 0}$,
such that $f(w0^k) = p_k$, where $w$ is empty or ends with 
$2_+$, i.e., the transmission probability cannot depend on anything
that happened \emph{before} the last collision.

\subsection{{\sc Avg}: Minimizing the Average Transmission Time}

Let $(p_k)_{k\ge 0}$ be the probability sequence 
corresponding to an optimal policy $f$ for {\sc avg}.
We first express our 
objective $\mathbb{E}X_1$ in terms of the sequence $(p_k)$.
Then, using the optimality of $f$, we deduce that $(p_k)_{k\ge 0}$ must
take on the special form described in Theorem~\ref{thm:min_average}.
This Theorem does not completely specify what the optimal protocol looks like.
Further calculations (Corollary~\ref{cor:average}) show that choosing 
$N=2$ is the best choice, and that {\bf {\sc avg}-Contention Resolution} 
(see Section~\ref{sect:introduction}) is an optimal protocol.

\begin{theorem}\label{thm:min_average}
There exists an integer $N>0$ 
and $a_0,a_1,a_2\in\mathbb{R}$ 
where $a_0-a_1+a_2=1$ and $a_0+a_1N+a_2N^2=0$ 
such that the following probability sequence
\begin{align*}
p_k=1-\frac{a_0+a_1k+a_2k^2}{a_0+a_1(k-1)+a_2(k-1)^2},\quad 0\leq k\leq N,
\end{align*}
induces an optimal policy that minimizes $\mathbb{E}X_1$. 
\end{theorem}
\begin{remark}
Note that defining $p_0,\ldots,p_N$ is sufficient,
since $p_N=1$ induces a certain collision if there are still 2 devices in
the system, 
which causes the algorithm to reset.  In the next time slot both devices 
would transmit with probability $p_0$.
\end{remark}
\begin{proof}
Assume we are using an optimal policy $f^*$ induced by 
a probability sequence $(p_i)_{i=1}^\infty$.
Define $S_1,S_2\ge 0$ to be the random variables indicating the 
\emph{index} of the first slot in which devices 1 and 2 first transmit.
Observe that $S_1$ and $S_2$ are i.i.d.~random variables,
where $\mathbb{P}(S_1=k)=\mathbb{P}(S_2=k)=p_k\prod_{i=0}^{k-1}(1-p_i)$.\footnote{We use the convention that $\prod_{i=0}^{-1}a_k=1$, 
where $(a_k)_{k=0}^\infty$ is any sequence.} 
We have
\begin{align}
\mathbb{E}X_1&=\sum_{k=0}^\infty\left[ \mathbb{P}(S_1=k)(k+1+\mathbb{P}(S_2= k)\cdot \mathbb{E}X_1)\right]\nonumber\\
\iff 
\mathbb{E}X_1&=\sum_{k=0}^\infty\left[ p_k \left(\prod_{i=0}^{k-1}(1-p_i)\right) \cdot \left(k+1+p_k\left(\prod_{i=0}^{k-1}(1-p_i)\right)\cdot \mathbb{E}X_1\right)\right].\label{eq:average_pk} 
\end{align}
Define $m_k=\prod_{i=0}^k (1-p_i)$ to be the probability that a device
idles in time steps $0$ through $k$, where $m_{-1} = 1$.
Note that $p_k m_{k-1}=m_{k-1}-m_{k}$ is true for all $k\geq 0$.
We can rewrite Eqn.~(\ref{eq:average_pk}) as:
\begin{align}
\mathbb{E}X_1&=\sum_{k=0}^\infty\left[ (m_{k-1}-m_k)(k+1+(m_{k-1}-m_k)\cdot \mathbb{E}X_1)\right]\nonumber\\
\iff \mathbb{E}X_1&=\mathbb{E}X_1\cdot \sum_{k=0}^\infty (m_{k-1}-m_k)^2+\sum_{k=0}^\infty (m_{k-1}-m_k)(k+1)\nonumber\\
\iff \mathbb{E}X_1&=\frac{\sum_{k=0}^\infty m_{k-1}}{1-\sum_{k=0}^\infty (m_{k-1}-m_k)^2}.\label{eq:average_mk}
\end{align}
By definition, $(m_k)_{k=-1}^\infty$ is a non-increasing 
sequence with $m_{-1}=1$ and $m_k\geq 0$. 
There is no optimal policy with 
$m_{k-1}=m_k\neq 0$ (meaning $p_k=0$), since otherwise we can delete 
$m_k$ from the sequence, 
leaving the denominator unchanged but reducing the numerator.
This implies $(m_k)_{k=-1}^\infty$ is either an infinite, 
positive, strictly decreasing sequence or a finite, positive,
strictly-decreasing sequence followed by a tail of zeros. 
Pick any index $k_0\ge 0$ such that 
$m_{k_0}>0$.  
We know $m_{k_0-1}>m_{k_0}>m_{k_0+1}$.  By the optimality of $f^*$, 
$m_{k_0}$ must, 
holding all other parameters fixed, be the optimal choice for 
this parameter in its neighborhood. In other words,
\begin{align*}
&\frac{\partial \mathbb{E}X_1}{\partial m_{k_0}}=0\\ \iff&\frac{1-\sum_{k=0}^\infty (m_{k-1}-m_k)^2+\sum_{k=0}^\infty m_{k-1}(-2(m_{k_0-1}-m_{k_0})+2(m_{k_0}-m_{k_0+1}))}{(1-\sum_{k=0}^\infty (m_{k-1}-m_k)^2)^2} = 0
\end{align*}
Therefore we have for any $k_0\geq 0$ such that $m_{k_0} > 0$,
\begin{align}
&2m_{k_0}-m_{k_0+1}-m_{k_0-1} = C \nonumber\\
\iff & m_{k_0}-m_{k_0+1}=m_{k_0-1}-m_{k_0} + C \label{eq:average_m_recur}
\end{align}
where $C=\frac{\sum_{k=0}^\infty (m_{k-1}-m_k)^2-1}{2\sum_{k=0}^\infty m_{k-1}}$ is a real constant. Note that $C=-\frac{1}{2\mathbb{E}X_1}<0$.
Fix any $k_1\ge 0$ such that $m_{k_1}>0$.  
By summing up Eqn.~(\ref{eq:average_m_recur}) 
for $k_0=0,1,\ldots,k_1$ and rearranging terms, we have
\begin{align}
m_{k_1}-m_{k_1+1}=(k_1+1) C +m_{-1}-m_{0}. \label{eq:average_m_recur1}
\end{align}
Fix any $k_2\ge 0$ such that $m_{k_2}>0$. 
By summing up Eqn.~(\ref{eq:average_m_recur1}) for $k_1=0,1,\ldots,k_2$, we have
\begin{align}
m_{k_2+1}&=(m_0-m_{-1})(k_2+2)+m_{-1}-\frac{(k_2+1)(k_2+2)}{2}C\label{eq:average_formula}\\
&=-\frac{C}{2}k_2^2 + \left(-\frac{3C}{2}+m_0-m_{-1}\right)k_2 + 2m_0-m_{-1}-C.
\end{align}
Recall that $C<0$ and $m_k\in[0,1]$. This rules out the possibility that the sequence $(m_{k})_{k=0}^\infty$ is an infinite strictly decreasing sequence, since a non-degenerate quadratic function is unbounded as $k$ goes to infinity. As a result, there must be a positive integer $N\ge 1$ for which 
$m_{N-1}>0$ and $m_N=0$. 
Also note that Eqn.~(\ref{eq:average_formula}) is not only true for $k_2=0,1,\ldots N-1$, but also true for $k_2=-1$ and $-2$. (This can be checked by directly setting $k_2=-1$ and $-2$.) 
We conclude that it is possible
to write $(m_k)$ as
\begin{align*}
m_k&={a_0+a_1k+a_2k^2},\quad -1\leq k\leq N,
\intertext{for some constants $a_0,a_1,a_2$ satisfying}
m_{-1} &= a_0 - a_1 + a_2 = 1\\
m_N &= a_0 + a_1N + a_2N^2 = 0.
\end{align*}
Writing $p_{k}=1-\frac{m_{k}}{m_{k-1}}$ gives the statement of the theorem.
\end{proof}

Based on Theorem \ref{thm:min_average}, we can find the optimal probability sequence for each fixed $N$ by choosing the best $a_2$. 
It turns out that $N=2$ is the best choice, 
though $N=3$ is only marginally worse.
The proof of Corollary \ref{cor:average} is in Appendix \ref{appendix:proof_average}.

\begin{corollary}\label{cor:average}
{\bf {\sc avg}-Contention Resolution} is an optimal protocol
for $n=2$ devices under the {\sc avg} objective.  The expected
average time is $\sqrt{3/2}+3/2 \approx 2.72474$.
\end{corollary}

\subsection{{\sc Min}: Minimizing the Earliest Transmission Time}

It is straightforward to show $\mathbb{E}\min(X_1,X_2) = 2$ under the optimal policy.  Nonetheless, it is useful to have a general closed form
expression for $\mathbb{E}\min(X_1,X_2)$ in terms of the $(m_k)$ sequence
of an arbitrary (suboptimal) policy, as shown in the proof of Theorem~\ref{thm:min_short}.  
This will come in handy later 
since $\mathbb{E}\max(X_1,X_2)$ can be expressed as $2\mathbb{E}X_1 - \mathbb{E}\min(X_1,X_2)$.

\begin{theorem}\label{thm:min_short}
The policy that minimizes $\mathbb{E}\min(X_1,X_2)$, 
{\bf {\sc min}-Contention Resolution}, transmits with constant
probability $1/2$ until successful.
Using the optimal policy, $\mathbb{E}\min(X_1,X_2) = 2$.
\end{theorem}
\begin{proof}
By Theorem~\ref{thm:recurrent} we can consider an optimal policy defined
by a sequence of transmission probabilities $(p_k)_{k\ge 0}$.
Let $H_{j,k}$ be the transmission/idle history of player $j\in \{1,2\}$ 
up to time slot $k$. Then we have
\begin{align}
\mathbb{E}\min(X_1,X_2) &=\sum_{k=0}^\infty \mathbb{P}(H_{1,k}=H_{2,k}=0^k1)(k+1+\mathbb{E}\min(X_1,X_2))\nonumber\\
&\;\;\;\;\;\;\;\; +\sum_{k=0}^\infty \mathbb{P}(\{H_{1,k},H_{2,k}\}=\{0^k1,0^k0\})(k+1)\nonumber\\
&=\sum_{k=0}^\infty\left( \prod_{i=0}^{k-1}(1-p_i)^2\cdot p_k^2(k+1+ \mathbb{E}\min(X_1,X_2))\right)\nonumber\\
&\;\;\;\;\;\;\;\; +\sum_{k=0}^\infty \left(\prod_{i=0}^{k-1}(1-p_i)^2\cdot 2p_k(1-p_k)(k+1)\right)\nonumber\\
&= \sum_{k=0}^\infty\left( \prod_{i=0}^{k-1}(1-p_i)^2\cdot \left((1-(1-p_k)^2)(k+1)+p_k^2\cdot \mathbb{E}\min(X_1,X_2)\right)\right).\nonumber
\intertext{Defining $m_k=\prod_{i=0}^k (1-p_i)$ as before, we have}
&= \sum_{k=0}^\infty \left((m_{k-1}^2 - m_k^2)(k+1) + m_{k-1}^2p_k^2\cdot \mathbb{E}\min(X_1,X_2)\right)\nonumber
\intertext{As $m_{k-1}p_k = m_{k-1}-m_k$, we can write $\mathbb{E}\min(X_1,X_2)$ in closed form as}
\mathbb{E}\min(X_1,X_2) &= \frac{\sum_{k=0}^\infty (m_{k-1}^2-m_k^2)(k+1)}{1-\sum_{k=0}^\infty(m_{k-1}-m_k)^2}\nonumber\\
&= \frac{\sum_{k=0}^\infty m_{k-1}^2}{1-\sum_{k=0}^\infty(m_{k-1}-m_k)^2}\label{eq:min_mk}\\
&=\frac{\sum_{k=0}^\infty m_{k-1}^2}{2\sum_{k=0}^\infty (m_{k-1}-m_k)m_k}\nonumber\\
&\geq \frac{\sum_{k=0}^\infty m_{k-1}^2}{2\sum_{k=0}^\infty \left(\frac{m_{k-1}}{2}\right)^2}=2
                    \;\;\;\;\;\; \mbox{($(m_{k-1}-m_k)m_k$ maximized when $m_k=m_{k-1}/2$.)}\nonumber
\end{align}
$\mathbb{E}\min(X_1,X_2)$ attains minimum $2$ if and only if for all $k\in\mathbb{N}$, $m_{k-1}-m_k=m_k$, i.e. $m_k=\frac{m_{k-1}}{2}$ and $m_{0}=\frac{1}{2}$. Thus $p_k=1-\frac{m_k}{m_{k-1}}=\frac{1}{2}$ for all $k$. This constant probability sequence corresponds to the constant policy with sending probability $\frac{1}{2}$ (i.e., {\bf {\sc min}-Contention Resolution}).
\end{proof}

\ignore{
Before going to find the optimal algorithm that minimizes the longer waiting time, we note done the following useful lemma.
\begin{lemma}\label{lem:short_property}
Let $X_1,X_2$ be the waiting times induced by the policy with constant sending probability $\frac{1}{2}$. Then $\mathbb{E}\min(X_1,X_2)=2$, $\mathbb{E}X_1=3$ and $\mathbb{E}\max(X_1,X_2)=4$.
\end{lemma}
\begin{proof}
By Theorem \ref{thm:min_short}, we know $\mathbb{E}\min(X_1,X_2)=2$. Note that since $\mathbb{E}\max(X_1,X_2)= \mathbb{E}(X_1+X_2)-\mathbb{E}\min(X_1,X_2)=2\mathbb{E}X_1-\mathbb{E}\min(X_1,X_2)$, we only need to calculate $\mathbb{E}X_1$.

We set $p_i=\frac{1}{2}$ in equation (\ref{eq:average_pk}) for all $i\in\mathbb{N}$, then we have 
\begin{align*}
    &\mathbb{E}X_1=\sum_{k=0}^\infty 2^{-k-1}(k+1+2^{-k-1}\mathbb{E}X_1)\\
    \iff & \mathbb{E}X_1=\frac{\sum_{k=0}^\infty 2^{-k-1}(k+1)}{1-\sum_{k=0}^\infty 4^{-k-1}}=3.
\end{align*}

Finally $\mathbb{E}\max(X_1,X_2)=2\mathbb{E}X_1-\mathbb{E}\min(X_1,X_2)=4$.
\end{proof}
}

\subsection{{\sc Max}: Minimizing the Last Transmission Time}

Before determining the optimal policy under the {\sc max} objective,
it is useful to have a crude estimate for its cost.
\begin{lemma}\label{lem:long_bound}
Let $f$ be the optimal policy for the {\sc max} objective and
$X_1^f,X_2^f$ be the latencies of the two devices.
Then $\mathbb{E}\max(X_1^f,X_2^f) \in [3,4]$.
\end{lemma}
\begin{proof}
The optimal policy under the {\sc min} objective, $f^*$, 
sends with probability 1/2 until successful.  It is easy to see
that $\mathbb{E}\max(X_1^{f^*}, X_2^{f^*}) = 4$, so $f$ can do no worse.
Under $f$ (or any policy), $\mathbb{E}\max(X_1^f,X_2^f) \ge 1 + \mathbb{E}\min(X_1^f,X_2^f)$.  By the optimality of $f^*$ for {\sc min}, 
$\mathbb{E}\min(X_1^f,X_2^f) \ge \mathbb{E}\min(X_1^{f^*},X_2^{f^*}) = 2$,
so $\mathbb{E}\max(X_1^f,X_2^f) \ge 3$.
\end{proof}

\begin{theorem}\label{thm:min_long}
Let $f$ be the optimal policy for the {\sc max} objective and define 
$1/\gamma = \mathbb{E}\max(X_1^f,X_2^f)$ to be its expected cost.
Let $x_1,x_2$ be the roots of the polynomial 
\begin{align}
x^2-(2-\gamma)x+1.\label{eq:x_definition}
\end{align}
There exists an integer $N\geq 0$ and reals 
$C_1,C_2$ where $C_1x_1^{-1}+C_2x_2^{-1}=0$ and $C_1x_1^{N+1}+C_2x_2^{N+1}=-1$, 
such that the following probability sequence
\begin{align*}
p_k=1-\frac{C_1x_1^k+C_2x_2^k+1}{C_1x_1^{k-1}+C_2x_2^{k-1}+1},\quad 0\leq k \leq N+1,
\end{align*}
induces an optimal policy that minimizes $\mathbb{E}\max(X_1,X_2)$.
\end{theorem}
\begin{remark}
Note that $p_{N+1}=1$, thus it is sufficient to only define $p_0,p_1,\ldots,p_{N+1}$.
\end{remark}
\begin{proof}
Assume the optimal policy $f$ is characterized by the probability sequence $(p_k)_{k=0}^\infty$. Using the derived expressions (Eqn.~(\ref{eq:average_mk}) and Eqn.~(\ref{eq:min_mk})) in Theorem \ref{thm:min_average} and \ref{thm:min_short}, we have 
\begin{align}
\mathbb{E}\max(X_1,X_2)&= 2\mathbb{E}X_1-\mathbb{E}\min(X_1,X_2)\nonumber\\
&=2\frac{\sum_{k=0}^\infty m_{k-1}}{1-\sum_{k=0}^\infty (m_{k-1}-m_k)^2}-\frac{\sum_{k=0}^\infty m_{k-1}^2}{1-\sum_{k=0}^\infty(m_{k-1}-m_k)^2}\nonumber\\
&=\frac{2\sum_{k=0}^\infty m_{k-1}-\sum_{k=0}^\infty m_{k-1}^2}{1-\sum_{k=0}^\infty (m_{k-1}-m_k)^2},\label{eq:longer_mk}
\end{align}
where $m_k=\prod_{i=0}^k (1-p_i)$ with $m_{-1}=1$.

The only requirement on the sequence $(m_k)_{k=-1}^\infty$ is that it is strictly decreasing with $m_{k}\in[0,1]$. 
First we observe if $m_{k}=m_{k+1}$, we must have both of them equal to zero. Otherwise, we can remove $m_{k}$ which will leave the denominator 
unchanged but reduce the numerator. 
Therefore, the optimal sequence is either a strictly decreasing sequence or a strictly decreasing sequence followed by a tail of zeros. 
Fix any $v\ge 0$ for which $m_v > 0$ we have, by the optimality of $(m_k)_{k=-1}^\infty$,
\begin{align}
\frac{\partial \mathbb{E}\max(X_1,X_2)}{\partial m_v}=\frac{(2-2m_v)B-2A(m_{v-1}-m_v-(m_v-m_{v+1}))}{D}=0,\label{eq:long_derivative}
\end{align}
where 
$B=1-\sum_{k=0}^\infty(m_{k-1}-m_k)^2$ and 
$A= 2\sum_{k=0}^\infty m_{k-1}  -\sum_{k=0}^\infty m_{k-1}^2$. 
Let $\gamma = \frac{B}{A}=\frac{1}{\mathbb{E}\max(X_1,X_2)}$, then we have, from Eqn.~(\ref{eq:long_derivative})
\begin{align}
&m_{v+1}=(2-\gamma)m_v-m_{v-1}+\gamma \label{eq:r}\\
\iff & (m_{v+1}-1) = (2-\gamma)(m_v-1)-(m_{v-1}-1) \label{eq:r1}
\end{align}
Eqn.~(\ref{eq:r1}) defines a linear homogeneous recurrence
relation for the sequence $(m_{v+1}-1)$, whose 
characteristic roots are $x_1,x_2=\frac{2-\gamma\pm\sqrt{\gamma^2-4\gamma}}{2}$.
One may verify that they satisfy the following identities.
\begin{align}
&x_1+x_2=2-\gamma \label{eq:long_sum_identity}\\
&x_1x_2=1 \label{eq:long_product_identity}
\end{align}
From Lemma \ref{lem:long_bound} we know $\gamma\in [\frac{1}{4},\frac{1}{3}]$. Thus we have $\gamma^2-4\gamma <0$ which implies $x_1$ and $x_2$ are distinct conjugate numbers and of the same norm $\sqrt{x_1x_2}=1$. 
Then $m_k -1 = C_1x_1^k + C_2x_2^k$ for all $k$ for which at least one of $m_{k-1},m_k$ or $m_{k+1}$ is greater than zero.

If it were the case that $m_k>0$ for all $k$, 
then by summing (\ref{eq:r}) up for all $v \in\mathbb{N}$, we have
\begin{align*}
\sum_{k=0}^\infty m_{k+1}=(2-\gamma)\sum_{k=0}^\infty m_k - \sum_{k=0}^\infty m_{k-1} + \gamma \cdot \infty
\end{align*}
which implies $\sum_{k=0}^\infty m_k=\infty$. This is impossible since, by
the upper bound of Lemma~\ref{lem:long_bound},
\[
4 \ge \mathbb{E}\max(X_1,X_2)=\frac{1+2\sum_{k=0}^\infty m_{k} -\sum_{k=0}^\infty m^2_{k}}{1-\sum_{k=0}^\infty (m_{k-1}-m_k)^2}\geq \frac{1+\sum_{k=0}^\infty m_{k}}{1}.
\]
Therefore the optimal sequence must be of the form
\begin{align*}
m_k =(C_1x_1^k + C_2x_2^k + 1)\mathbbm{1}_{k\leq N}
\end{align*}
for some integer $N\geq 0$, where $C_1x_1^N+C_2x_2^N+1=0$ and $C_1x_1^{-1}+C_2x_2^{-1}+1=1$. Writing $p_k=1-\frac{m_k}{m_{k-1}}$ gives the statement of the theorem.
\end{proof}

The proof of Corollary \ref{cor:long} is in Appendix \ref{appendix:proof_long}.

\begin{corollary}\label{cor:long}
{\bf {\sc max}-Contention Resolution} 
is an optimal protocol
for $n=2$ devices under the {\sc max} objective.  The expected maximum latency is 
$1/\gamma \approx 3.33641$,
where $\gamma$ is the unique root of
$3x^3 - 12x^2 + 10x - 2$ in the interval $[1/4,1/3]$.
\end{corollary}

\section{Conclusion}
In this paper we established the existence of optimal contention resolution policies for any \emph{reasonable} cost metric,
and derived the first optimal protocols for resolving conflicts
between $n=2$ parties under the {\sc avg}, {\sc min}, and {\sc max}
objectives.

Generalizing our results to $n\ge 3$ or to more complicated cost metrics
(e.g., the $\ell_2$ norm) is a challenging problem.  Unlike the $n=2$ case,
it is not clear, for example, whether the optimal protocols for $n=3$
select their transmission probabilities from a finite set of reals.
It is also unclear whether the optimal protocols for $n=3$ 
satisfy some analogue of Theorem~\ref{thm:recurrent}, i.e., that
they are ``recurrent'' in some way.

\bibliographystyle{plain}
\bibliography{refs}

\newpage
\begin{appendices}

\section{Proof of Theorems}
\label{appendix:proof}
\subsection{Deduction on Random Board Model}

In order to implement a contention resolution policy we need
to generate biased random bits, e.g., in order 
to transmit with probability $1/3$.
However, the bias of each time step could be different and depend
on the outcome of previous time steps.  It will be convenient if we 
can generate all randomness in advance and dynamically set the biases as we go.
To that end we define the Random Board model, which is just an infinite number
of uniform and i.i.d.~random reals in $[0,1]$.  If device $k$ in time step $t$ wants to generate a biased random bit $b$ with probability $1/3$ of 1,
it sets $b = \mathbbm{1}(U_{k,t} < 1/3)$, 
where $(U_{k,t})$ is the random board.  

\begin{definition}[Random Board]
A random board $U$ with $n$ rows are a set of $\textrm{i.i.d.}$ uniformly distributed random variables $(U_{k,t})_{k\in[n], t\in\mathbb{N}}$ 
with range $[0,1]$.
\end{definition}

\begin{definition}[Deduction on Random Board]
Given a policy $f\in \mathcal{F}$ and an outcome $(u_{k,t})_{k\in[n],t\in\mathbb{N}}$ of an $n$-row random board where $u_{k,t}\in[0,1]$, we deduce $(\bar{d}_{k,t},\bar{r}_{k,t})_{k\in[n],t\in\mathbb{N}}$ iteratively by the following rule.
\begin{align*}
\bar{d}_{k,t}&=\mathbbm{1}(u_{k,t}< f((\bar{r}_{k,0}\bar{r}_{k,1}\ldots\bar{r}_{k,t-1}))) & \mbox{(1 iff device $k$ transmits at time $t$)}\\
\bar{r}_{k,t}&= \begin{cases}
0,&\bar{d}_{k,t}=0\\
1,&(\bar{d}_{k,t}=1) \land (\forall j\neq k, \bar{d}_{j,t}=0)\\
2_+,&(\bar{d}_{k,t}=1) \land (\exists j\neq k, \bar{d}_{j,t}=1)
\end{cases} & \mbox{(the response at time $t+0.5$)}
\end{align*}
We define $\bar{D}_{k,t}$ and $\bar{R}_{k,t}$ be the random variables that maps outcomes to the deduced $\bar{d}_{k,t}$ and $\bar{r}_{k,t}$ respectively. 
\end{definition}

We demonstrate the deduction on random board by the following example.
\begin{example}
Let Table \ref{tab:outcome3} be an outcome $w_0$ of a 3-row random board.

\begin{table}[h!]
\centering
\begin{tabular}{ccccccc}
 0.23371 & 0.281399 & 0.375409 & 0.927202 & 0.0824814 & 0.0473227 &\ldots \\
 0.216321 & 0.4534 & 0.377702 & 0.573771 & 0.704855 & 0.497943 &\ldots \\
 0.888769 & 0.939998 & 0.261829 & 0.343283 & 0.830001 & 0.43118 &\ldots \\
\end{tabular}
\caption{An outcome $w_0$ of a 3-row random board}
\label{tab:outcome3}
\end{table}

Let $f_1$ be the policy with constant sending probability $\frac{1}{2}$ until the message gets successfully transmitted. Let $f_2$ be the policy with constant sending probability $\frac{1}{3}$ before success. Then the results of deduction are shown in the following two tables $(\bar{d}_{k,t},\bar{r}_{k,t})_{k\in [n],t\in\mathbb{N}}$ (successful transmissions are marked by a star $*$).
\begin{table}[h!]
\centering
\begin{tabular}{ccccccc}
 (send,$2^+$) & (send,$2^+$) & (send,$2^+$) & (idle,0) & (send*,1) & (idle,0) &\ldots \\
 (send,$2^+$) & (send,$2^+$) & (send,$2^+$) & (idle,0) & (idle,0) & (send*,1) &\ldots \\
 (idle,0) & (idle,0) & (send,$2^+$) & (send*,1) & (idle,0) & (idle,0) &\ldots \\
\end{tabular}
\caption{Deduction on $w_0$ using policy $f_1$ (constant sending probability $\frac{1}{2}$)}
\end{table}

\begin{table}[h!]
\centering
\begin{tabular}{ccccccc}
 (send,$2^+$) & (send*,1) & (idle,0) & (idle,0) & (idle,0) & (idle,0) &\ldots \\
 (send,$2^+$) & (idle,0) & (idle,0) & (idle,0) & (idle,0) & (idle,0) &\ldots \\
 (idle,0) & (idle,0) &  (send*,1)  & (idle,0) & (idle,0) & (idle,0) &\ldots \\
\end{tabular}
\caption{Deduction on $w_0$ using policy $f_2$ (constant sending probability $\frac{1}{3}$)}
\end{table}
\end{example}

\begin{lemma}\label{lem:couple}
Given an $n$-row random board $U$, for any policy $f\in\mathcal{F}$, the deduced set of random variables $(\bar{D}_{k,t},\bar{R}_{k,t})_{k\in [n],t\in\mathbb{N}}$ satisfies both equations (\ref{eq:send_prob}) and (\ref{eq:response}).
\end{lemma}
\begin{proof}
Equation (\ref{eq:response}) follows directly by the definition of deduction. For equation (\ref{eq:send_prob}), we have
\begin{align*}
\mathbb{P}(\bar{D}_{k,t}=1\mid \bar{R}_{k,1}\bar{R}_{k,2}\ldots\bar{R}_{k,t}=h)&=\mathbb{P}(U_{k,t}< f(h))\\
&=f(h).
\end{align*}
\end{proof}

From Lemma \ref{lem:couple} we know, using the same $f$, the random processes $(D_{k,t},R_{k,t})_{k\in[n],t\in\mathbb{N}}$ and $(\bar{D}_{k,t},\bar{R}_{k,t})_{k\in[n],t\in\mathbb{N}}$ are of identical distributions and thus $(X_i)_{i\in[n]}$ and $(\bar{X}_i)_{i\in[n]}$ are also identically distributed. Having this, we now can consider all the policies working on a same sample space, i.e. the sample space of the random board. 

\subsection{Existence Issues}\label{appendix:proof_existence}
For notation, we denote any random variable $X$ (e.g. $X_i$, $D_{k,t}$ and $R_{k,t}$) induced by policy $f$ as $X^f$. 

Having the random board model set up, we now can go on proving the existence theorem. We first formally define the class of \emph{reasonable} objective functions.

\begin{definition}
We say an objective function $T:\mathbb{R}^n\to\mathbb{R}^+$ is \emph{reasonable} if for any $s>0$, there exists some $N_s>0$ such that $\mathbbm{1}(T(X_1,X_2,\ldots,X_n)<s)$ only depends on the deduction of the first $N_s$ time slots. In other words, there exists a function $h$ such that $\mathbbm{1}(T(X_1,X_2,\ldots,X_n)<s)=h(\{D_{k,t},R_{k,t}\}_{k\in[n],t\in[N_s]})$.
\end{definition}

\begin{proof}[Proof of Theorem \ref{thm:existence}]
Let $g:\mathcal{F}\to \mathbb{R}_+\cup\{\infty\}$ be the function that maps policy $f$ to $\mathbb{E}T(X_1^f,X_2^f,\ldots,X_n^f)$. For simplicity, we denote $(X_1^f,X_2^f,\ldots,X_n^f)$ by $X^f$.

Define $J$ as $\inf\{g(f)\mid f\in\mathcal{F}\}$ which exists since $g$ is bounded below ($g$ is non-negative). Since the set of finite strings is countable, we can identify $\mathcal{F}$ by $[0,1]^\mathbb{N}$. Note that $[0,1]$ is compact and thus, by diagonalization argument, we can find a sequence of $(f_i)_{i\in\mathbb{N}}$ where $f_i\in\mathcal{F}$ converges to $f^*$ point-wisely with $\lim_{i\to \infty} g(f_i)=J$.  We have to show $g(f^*)=J$. By definition, we have
\begin{align*}
g(f^*)&=\mathbb{E}T(X^{f^*})\\
&=\int_{0}^\infty \mathbb{P}(T(X^{f^*})\geq s)ds\\
&=\int_{0}^\infty (1-\mathbb{P}(T(X^{f^*})< s))ds
\end{align*}
For a fixed $s>0$, by the definition of a reasonbale objective function, there exists a constant $N_s>0$ such that $\mathbbm{1}(T(X^f)<s)$ only depends on the deduction of the first $N_s$ time slots. Since $f_j$ converges to $f^*$ point-wisely and the set of strings $\{w\mid |w|<N_s\}$ is finite\footnote{$|w|$ is equal to the length of the string $w$.}, we can choose $\epsilon$ and $N$ so that for any $j>N$ and $w$ with $|w|<N_s$, we have $|f_j(w)-f^*(w)|<\epsilon$. Then for the two policies $f_j$ and $f^*$, the probability that the deduction of the first $N_s$ time slots will not change is larger than $(1-\epsilon)^{nN_s}$. In fact, let $w$ be the history of a user with $|w|<N_s$ and $U$ be the cell on the random board that will be used by the user, we have $\mathbb{P}(\mathbbm{1}(U<f(w))=\mathbbm{1}(U<f^*(w)))=1-|f(w)-f^*(w)|\geq (1-\epsilon)$. Therefore, $\mathbb{P}(E_k\mid E_{k-1})\geq (1-\epsilon)^n$, where $E_k$ is the event that the deductions of the first $k$ columns are identical. We get the estimation we want by writing out
\begin{align*}
    \mathbb{P}(E_{N_s-1})&=\mathbb{P}(E_{N_s-1}\mid E_{N_{s}-2})\mathbb{P}(E_{N_{s}-2}\mid E_{N_s-3})\cdots \mathbb{P}(E_1\mid E_0)\mathbb{P}(E_0)\\
    &\geq (1-\epsilon)^{n N_s}
\end{align*}

Thus, let $\Omega$ be the sample space of the random board,
\begin{align*}
    |\mathbb{P}(T(X^{f_j})<s)-\mathbb{P}(T(X^{f^*})<s)| &=  \left|\int_{ \Omega}\mathbbm{1}(T(X^{f_j}(w))<s)-\mathbbm{1}(T(X^{f^*}(w))<s) dw\right|\\
    &\leq \int_{\Omega}\left|\mathbbm{1}(T(X^{f_j}(w))<s)-\mathbbm{1}(T(X^{f^*}(w))<s) \right|dw\\
    &=\mathbb{P}\left(\mathbbm{1}\left((T(X^{f_j})<s\right)\neq\mathbbm{1}\left(T(X^{f^*})<s)\right)\right)\\
    &= 1-\mathbb{P}\left(\mathbbm{1}\left((T(X^{f_j})<s\right)=\mathbbm{1}\left(T(X^{f^*})<s)\right)\right)\\
    &\leq 1-\mathbb{P}(E_{N_s})\\
    &\leq 1-(1-\epsilon)^{nN_s}
\end{align*} 
which can be made arbitrarily small by choosing $\epsilon$ small enough. Thus we have $\mathbb{P}(T(X^{\lim_{j\to\infty}f_j})<s)=\lim_{j\to \infty}\mathbb{P}(T(X^{f_j})<s)$. Therefore we have
\begin{align*}
g(f^*)&=\int_{0}^\infty (1-\mathbb{P}(T(X^{f^*})< s))ds\\
&=\int_{0}^\infty (1-\lim_{j\to \infty}\mathbb{P}(T(X^{f_j})< s))ds\\
&=\int_{0}^\infty\lim_{j\to \infty}\mathbb{P}(T(X^{f_j})\geq s)ds\\	
&\leq\lim_{j\to \infty}\int_{0}^\infty \mathbb{P}(T(X^{f_j})\geq s)ds\\
& = J
\end{align*}
Note that we have used the Fatou's Lemma~\cite{royden1988real} to change the order of integration and the limit. We have $g(f^*)\leq J$ but $J$ is the infimum. Thus $f^*$ is an optimal policy.

\end{proof}

\subsection{Optimal Recurrent Policies --- Proof of Theorem~\ref{thm:recurrent}}\label{appendix:proof_recurrent}

\begin{proof}[Proof of Theorem \ref{thm:recurrent}]
Let $C$ be the random variable indicating the time of the first collision. We define $C=-1$ if there is no collision. By the definition of $f^*$, we know $C^{f^*}=C^f$ since $f^*(w)=f(w)$ if $2_+\not\in w$. Also, if there is no collision, the deduction of the two policies should be identical. Thus we have $\mathbb{E}(T(X_1^f,X_2^f)\mid C^f=-1)=\mathbb{E}(T(X_1^{f^*},X_2^{f^*})\mid C^{f^*}=-1)=\vcentcolon M$. Let $\mathbb{P}(C^f=k)=\mathbb{P}(C^{f^*}=k)=q_k$ and $f_k(w)=f(0^k2_+w)$ for any $w\in\mathcal{R}^*$. Note that both users must keep idle before the first collision (if there is a collision). Therefore if the first collision happens at time $k$, after time $k$, the two users would behave as if they have restarted the process with a new policy $f_k$. Then we have,
\begin{align*}
    \mathbb{E}T(X_1^f,X_2^f)&=Mq_{-1} + \sum_{k=0}^\infty \mathbb{E}(T(X_1^{f},X_2^{f})\mid C^f=k)q_k\\&=Mq_{-1} + \sum_{k=0}^\infty \mathbb{E}T(k+1+X_1^{f_k},k+1+X_2^{f_k})q_k\\
    &=Mq_{-1} +\sum_{k=0}^\infty (k+1)q_k+ \sum_{k=0}^\infty \mathbb{E}T(X_1^{f_k},X_2^{f_k})q_k\\
    &\geq Mq_{-1} +\sum_{k=0}^\infty (k+1)q_k+ \mathbb{E}T(X_1^{f},X_2^{f})\sum_{k=0}^\infty q_k
\end{align*}
Note that if $\sum_{k=0}^\infty q_k =1$, then $q_{-1}=0$ and we would have $\sum_{k=0}^\infty(k+1)q_k\leq 0$, which is impossible. Thus we must have $\sum_{k=0}^\infty q_k<1$. 

Now by the definition of $f^*$, we have $f^*(w)=f^*(0^k2_+w)$. Similarly, we have
\begin{align*}
    \mathbb{E}T(X_1^{f^*},X_2^{f^*})
    &=Mq_{-1} + \sum_{k=0}^\infty \mathbb{E}(T(X_1^{f^*},X_2^{f^*})\mid C^{f^*}=k)q_k\\
    &=Mq_{-1} + \sum_{k=0}^\infty \mathbb{E}T(k+1+X_1^{f^*},k+1+X_2^{f^*})q_k\\
    &=Mq_{-1} +\sum_{k=0}^\infty (k+1)q_k+ \sum_{k=0}^\infty \mathbb{E}T(X_1^{f^*},X_2^{f^*})q_k
\end{align*}

We conclude that $\mathbb{E}T(X_1^f,X_2^f)\geq \frac{Mq_{-1} +\sum_{k=0}^\infty (k+1)q_k}{1-\sum_{k=0}^\infty q_k}=\mathbb{E}T(X_1^{f^*},X_2^{f^*})$. Since $f$ is optimal, then $f^*$ is also optimal.
\end{proof}

\subsection{Proof of Corollary~\ref{cor:average} --- 
The {\sc Avg} Objective}\label{appendix:proof_average}

\begin{proof}[Proof of Corollary \ref{cor:average}]
By Theorem \ref{thm:min_average}, the optimal 
sequence $(m_k)_{k=-1}^\infty$ is given by $m_k=(a_0+a_1k+a_2k^2)\mathbbm{1}_{k\leq N}$, for some integer $N>0$,
where 
\begin{align}
a_0-a_1+a_2&=1\label{eq:avg_condition_1}\\
a_0+a_1N+a_2N^2&=0\label{eq:avg_condition_2}
\end{align}
Note that by the definition, we must have $m_{k-1}\leq m_k$ for all $k=0,1,\ldots, N$. We necessarily have 
\begin{align}
    &m_{-1} \geq m_0 \implies a_0-a_1+a_2\geq a_0 \label{eq:a2_condition_1}\\
    &m_{N-1} \geq m_N \implies a_0+a_1(N-1)+a_2(N-1)^2\geq a_0+a_1 N +a_2N^2 \label{eq:a2_condition_2}
\end{align}
From Eqn.~(\ref{eq:avg_condition_1}) and (\ref{eq:avg_condition_2}), we have $a_0=\frac{N-a_2(N^2+N)}{N+1}$ and $a_1=\frac{-a_2(N^2-1)-1}{N+1}$. Insert both expressions into Eqn.~(\ref{eq:a2_condition_1}) and (\ref{eq:a2_condition_2}), we get the constraints on $a_2$:
\begin{align}
a_2\in \left[-\frac{1}{N+N^2},\frac{1}{N+N^2}\right].\label{eq:a2_constraint}
\end{align}
In fact, the conditions that $m_{-1}\geq m_0$ and $m_{N-1}\geq m_N$ are sufficient for $m_{k-1}\leq m_k$ for all $k=0,1,\ldots,N$. To see this, let $g(x)=a_0+a_1x+a_2x^2$ be the quadratic function where $g(k)=m_k$ for all $k=-1,0,1,\ldots,N$. Once we know $g(-1)\geq g(0)$ and $g(N-1)\geq g(N)$, we know there exist $x_1\in [-1,0]$ and $x_2\in [N-1,N]$ such that $g'(x_1)\leq 0$ and $g'(x_2)\leq 0$. Since $g'(x)=a_1+2a_2x$ is a linear function, $g'(x_1),g'(x_2)\leq 0$ implies $g'(x)\leq 0$ for any $x\in[x_1,x_2]$. Thus $g'(x)\leq 0$ for any $x\in[0,N-1]$. We conclude that $m_{k-1}=g(k-1)\leq g(k)=m_k$ is also true for any $k=1,\ldots,N-1$.

Using Eqn.~(\ref{eq:average_mk}) derived in Theorem \ref{thm:min_average}, we have
\begin{align}
\mathbb{E}X_1&=\frac{\sum_{k=-1}^{N-1} (a_0+a_1k+a_2k^2)}{1-\sum_{k=0}^N(-a_1+a_2(-2k+1))^2}\nonumber\\
&=\frac{(N+1) (N+2) (a_2 N (N+1)-3)}{2 N \left(a_2^2 (N+1)^2 (N+2)-3\right)}\label{eq:avg_a2_N}
\end{align}
Fixing $N$, to minimize $\mathbb{E}X_1$, we have
\begin{align*}
\frac{\partial \mathbb{E}X_1}{\partial a_2}=-\frac{(N+1)^2 (N+2) (a_2^2 N (N+2) (N+1)^2-6 a_2 (N+2) (N+1)+3 N)}{2 N \left(a_2^2 (N+1)^2 (N+2)-3\right)^2}=0
\end{align*}
Then $a_2=\frac{3 N^2+9 N+6\pm\sqrt{3} \sqrt{-N^5-N^4+13 N^3+37 N^2+36 N+12}}{N^4+4 N^3+5 N^2+2 N}.$
It is easy to check that $a_2$ becomes complex for all $N\geq 5$, 
which implies $\frac{\partial \mathbb{E}X_1}{\partial a_2}\leq 0$ for any $a_2$.  That is to say for $N\geq 5$, the larger $a_2$ is, the smaller $\mathbb{E}X_1$ is.
Thus the best $a_2$ one can pick under the constraint (\ref{eq:a2_constraint}) is $\frac{1}{N+N^2}$. In this case, by inserting $a_2=\frac{1}{N+N^2}$ into Eqn.~(\ref{eq:avg_a2_N}) we have
\begin{align}
\mathbb{E}X_1&=\frac{(N+1) (N+2) (1-3)}{2 N \left( \frac{N+2}{N^2}-3\right)}\nonumber\\
&=\frac{N(N+1)(N+2)}{3N^2-N-2}.\label{eq:avg_N}
\end{align}
Then we differentiate Eqn.~(\ref{eq:avg_N}) by $N$:
\begin{align*}
\frac{\partial \mathbb{E}X_1}{\partial N} = \frac{1}{15} \left(-\frac{8}{(3 N+2)^2}-\frac{18}{(N-1)^2}+5\right)>0
\end{align*}
for $N\geq 5$. Therefore, in the case where $N\geq 5$, it it best
to choose $N=5$, yielding
$\mathbb{E}X_1=\frac{N(N+1)(N+2)}{3N^2-N-2}=\frac{105}{34}\approx 3.09$.

Now we consider the cases where $N\in\{1,2,3,4\}$. We give the optimal results for each case in the Table \ref{tab:average_solution}. Note that we have $a_2\in[-\frac{1}{N+N^2},\frac{1}{N+N^2}]$ for each $N$.
\begin{table}[h!]
\centering
\begin{tabular}{cccc}
$N$ & $\mathbb{E}X_1$ (Eqn.~(\ref{eq:avg_a2_N})) & optimal $a_2$ & optimal $\mathbb{E}X_1$ \\
\hline\\[-10pt]
$1$ & $\displaystyle\frac{3-2 a_2}{1-4 a_2^2}$ & $\displaystyle\frac{3}{2}-\sqrt{2}$ & $\displaystyle\sqrt{2}+\frac{3}{2}\approx 2.91421$ \\[15pt]
$2$ & $\displaystyle\frac{3-6 a_2}{1-12 a_2^2}$ & $\displaystyle\frac{1}{2}-\frac{1}{\sqrt{6}}$ & $\displaystyle\frac{1}{2} (\sqrt{6}+3)\approx 2.72474$ \\[15pt]
$3$ & $\displaystyle\frac{10-40 a_2}{3-80 a_2^2}$ & $\displaystyle\frac{1}{12}$ & $\displaystyle\frac{30}{11}\approx 2.72727$ \\[15pt]
$4$ & $\displaystyle\frac{5 \left(20 a_2-3\right)}{200 a_2^2-4}$ & $\displaystyle\frac{1}{20}$ & $\displaystyle \frac{20}{7}\approx 2.85714$   
\end{tabular}
\caption{Solutions for $N\in\{1,2,3,4\}$}
\label{tab:average_solution}
\end{table}

Combining all previous results we conclude that a globally optimal algorithm that minimizes the average waiting time of a two-party contention is obtained when $N=2$ and $a_2=\frac{1}{2}-\frac{1}{\sqrt{6}}$. In this case, we have $a_0=\frac{N-a_2(N^2+N)}{N+1}=\frac{1}{3} (\sqrt{6}-1)$ and $a_1=\frac{-a_2(N^2-1)-1}{N+1}=\frac{\sqrt{6}-5}{6}$. Then we have $m_0=a_0=\frac{1}{3} (\sqrt{6}-1)$, $m_1=a_0+a_1+a_2=\frac{1}{3} (\sqrt{6}-2)$ and $m_2=0$. Finally we get the optimal probability sequence.
\begin{align*}
&p_0=1-m_0=\frac{4-\sqrt{6}}{3}\approx 0.516837,\\
&p_1=1-\frac{m_1}{m_0}=\frac{1+\sqrt{6}}{5}\approx 0.689898,\\
&p_2=1-\frac{m_2}{m_1}=1.
\end{align*}
This is precisely the {\bf {\sc avg}-Contention Resolution} protocol.
\end{proof}

\subsection{Proof of Corollary~\ref{cor:long} --- The {\sc Max} Objective}\label{appendix:proof_long}

\begin{proof}[Proof of Corollary \ref{cor:long}]
We continue to use all the derived results in Theorem \ref{thm:min_long}. Remember that $x_1$ and $x_2$ depend on $\gamma$, and satisfy the 
equalities $x_1+x_2 = 2-\gamma$ and $x_1x_2=1$.  
Recall that $1/\gamma= \mathbb{E}\max(X_1,X_2)$.
Using the fact that $m_{-1}=1$, it follows from Eqn.~(\ref{eq:longer_mk}) 
that $\gamma$ can be written as
\begin{align*}
\gamma  &=\frac{1-\sum_{k=0}^{N+1} (m_{k-1}-m_k)^2}{1+2\sum_{k=0}^N m_{k} -\sum_{k=0}^N m^2_{k}}\\
&=\frac{1-(1-m_0)^2-\sum_{k=0}^{N} \left(C_1x_1^{k}(1-x_1)+C_2x_2^{k}(1-x_2)\right)^2}{1+2\sum_{k=0}^N \left(C_1x_1^k+C_2x_2^k+1\right) -\sum_{k=0}^N \left(C_1x_1^k+C_2x_2^k+1\right)^2}\\
\intertext{Since $x_1x_2=1$, this can be written as follows.}
&=\frac{1-(1-m_0)^2-\sum_{k=0}^{N} \left(C_1^2\left(1-x_1\right)^2\left(x_1^2\right)^k+2C_1C_2(1-x_1)(1-x_2)+C_2^2(1-x_2)^2\left(x_2^2\right)^{k}\right)}{N+2 -\sum_{k=0}^N \left(C_1^2\left(x_1^2\right)^k+C_2^2\left(x_2^2\right)^k+2C_1C_2\right)}
\end{align*}
By organizing the terms we have (using identities (\ref{eq:long_sum_identity}), (\ref{eq:long_product_identity}), the polynomial (\ref{eq:x_definition}) that defines $x_1$ and $x_2$ and the definition $m_k = C_1x_1^k + C_2x_2^k + 1$)
\begin{align*}
0 &=\gamma(N+2)-1+(1-m_0)^2-2(N+1)C_1C_2(-(1-x_1)(1-x_2)+\gamma) \\
&\hspace{1cm} +C_1^2\left((1-x_1)^2-\gamma\right)\sum_{k=0}^N \left(x_1^2\right)^k + C_2^2\left((1-x_2)^2-\gamma\right)\sum_{k=0}^N \left(x_2^2\right)^k
\intertext{Using the fact that
$(1-x_1)(1-x_2)=2-(x_1+x_2)=\gamma$
and that 
$-\gamma(1+x_1)=(x_1+x_2-2)(1+x_1) 
= x_1^2 - x_1 - 1 + x_2 = (1-x_1)^2 - 2 + x_1 + x_2
= (1-x_1)^2 - \gamma$, we can simplify this as follows.}
&=\gamma(N+2)-1+(1-m_0)^2 -C_1^2\gamma(1+x_1)\sum_{k=0}^N \left(x_1^2\right)^k  - C_2^2\gamma(1+x_2)\sum_{k=0}^N \left(x_2^2\right)^k\\
&=\gamma(N+2)-1+(1-m_0)^2-C_1^2\gamma\frac{1-\left(x_1^2\right)^{N+1}}{1-x_1} - C_2^2\gamma\frac{1-\left(x_2^2\right)^{N+1}}{1-x_2}
\intertext{We will now put the last two terms
under the common denominator $(1-x_1)(1-x_2) = \gamma$,
thereby cancelling the $\gamma$ factors in both terms.}
&=\gamma(N+2) -1+(1-m_0)^2\\
&\hspace{1cm} -\left(C_1^2+C_2^2-C_1^2x_2-C_2^2x_1-C_1^2x_1^{2(N+1)}-C_2^2x_2^{2(N+1)}+C_1^2x_1^{2N+1}+C_2^2x_2^{2N+1}\right)
\intertext{Since $m_0 = C_0 + C_1 + 1$, we have}
&=\gamma(N+2)-1 +2C_1C_2\\
&\hspace{1cm} -(-C_1^2x_2-C_2^2x_1-C_1^2x_1^{2(N+1)}-C_2^2x_2^{2(N+1)}+C_1^2x_1^{2N+1}+C_2^2x_2^{2N+1})\\
&=\gamma(N+2)-1 +(m_{N+1}-1)^2
-(-C_1^2x_2-C_2^2x_1+C_1^2x_1^{2N+1}+C_2^2x_2^{2N+1})\\
&=\gamma(N+2)-1 +(m_{N+1}-1)^2-(m_N-1)(m_{N+1}-1)+(m_0-1)(m_{-1}-1)
\intertext{Note that we have $m_{-1}=1$ and $m_{N+1}=0$, hence}
&=\gamma(N+2) + m_N-1.
\end{align*}
Rearranging terms, we have:
\begin{align}
N &= \frac{1-m_N}{\gamma}-2 < 2\label{eq:gamma_N_relation}
\end{align}
The last inequality follows from the fact that
$m_N\in(0,1]$ and $\frac{1}{\gamma}\in[3,4]$. 
Thus the only possible values of $N$ are $0$ and $1$.

If $N=0$, we have the following set of equations. \footnote{We have used the algebra software \emph{Mathematica} to give the exact form of the solutions for both cases.}
\begin{align*}
\gamma(N+2) + m_N -1 = 2\gamma + C_1+C_2 &= 0 \\
C_1x_1^{-1}+C_2x_2^{-1}+1 &=1   &(m_{-1}=1)\\
C_1 x_1+C_2 x_2+1   &=0         &(m_1=0)
\end{align*}
whose solution is $\gamma =\frac{2-\sqrt{2}}{2}$ where $\frac{1}{\gamma}\approx 3.41421$.

If $N=1$, we have the following set of equations.
\begin{align}
\gamma(N+2) + m_N -1 = 3\gamma + C_1x_1+C_2x_2 &= 0\\ 
C_1x_1^{-1}+C_2x_2^{-1} + 1 &=1   &(m_{-1}=1)\label{eq:final_2}\\
C_1 x_1^2+C_2 x_2^2 + 1&=0    &(m_2=0)\label{eq:final_3}
\end{align}
whose solution gives $$\gamma \approx 0.299723, \text{ where } 3 \gamma^3-12 \gamma^2+10 \gamma-2=0 $$
We see in this case $\frac{1}{\gamma}\approx 3.33641$, which is better than the case where $N=0$. The corresponding $C_1\approx -0.264419-0.426908 i $ and $C_2 \approx -0.264419+0.426908 i$, which are a pair of conjugate roots of polynomial $76 x^6-532 x^5+664 x^4+3288 x^3+4680 x^2+2268 x+729$.

Finally we get 
\begin{align*}
&p_0 = 1- m_0=-C_1-C_2=\alpha\approx 0.528837 , \text{ where } \alpha^3+7 \alpha^2-21 \alpha+9=0,\\
&p_1 = 1-\frac{m_1}{m_0}=1-\frac{C_1x_1+C_2x_2+1}{C_1+C_2+1}=\beta \approx 0.785997, \text{ where } 4\beta^3-8\beta^2+3=0,\\
&p_2 = 1-\frac{m_2}{m_1}=1.
\end{align*}
This is precisely the {\bf {\sc max}-Contention Resolution} protocol.
\end{proof}
\end{appendices}

\end{document}